\documentclass[11pt,conference]{article}

\usepackage{amsmath}
\usepackage{amsthm}
\usepackage{latexsym}
\usepackage{amssymb}
\usepackage{fullpage}
\usepackage{verbatim}
\usepackage{color}
\usepackage[dvipsnames]{xcolor}
\usepackage{graphicx}

\newtheorem{theorem}{Theorem}
\newtheorem{remark}{Remark}
\newtheorem{lemma}{Lemma}

\def\mbf{\mathbf}
\def\mc{\mathcal}

\newcommand{\bigO}{\mc{O}}

\newcommand{\E}{\mathbb{E}}
\renewcommand{\P}{\mathbb{P}}
\newcommand{\bern}{\mathsf{Bernoulli}}
\newcommand{\wt}{\mathop{\mathrm{wt}_{\mathrm{H}}}\nolimits}
\newcommand{\dH}{d_{\mathrm{H}}}

\newcommand{\cX}{\mc{X}}
\newcommand{\cY}{\mc{Y}}
\newcommand{\cZ}{\mc{Z}}

\newcommand{\cC}{\mc{C}}
\newcommand{\cT}{\mc{T}}
\newcommand{\cTp}{{{\mc{T}}'}}

\newcommand{\bX}{\mbf{X}}
\newcommand{\bY}{\mbf{Y}}
\newcommand{\bZ}{\mbf{Z}}
\newcommand{\be}{\mbf{e}}
\newcommand{\bu}{\mbf{u}}
\newcommand{\bx}{\mbf{x}}
\newcommand{\by}{\mbf{y}}
\newcommand{\bz}{\mbf{z}}
\newcommand{\e}{\varepsilon}

\newcommand{\renc}{\Phi}
\newcommand{\ddec}{\psi}

\newcommand{\delay}{\Delta}
\newcommand{\cost}{c}
\newcommand{\Perr}{P_{\mathsf{err}}}
\newcommand{\jamset}{\mc{G}}

\newcommand{\er}{\perp}
\newcommand{\uner}{u}
\newcommand{\capgap}{\epsilon}
\newcommand{\decslack}{\capgap/2}
\newcommand{\setcoh}{\eta_1}
\newcommand{\codecoh}{\eta_2}
\newcommand{\cohbound}{\eta_3}
\newcommand{\decT}{ \tau }
\newcommand{\chunk}[2]{ S(#1,#2) }

\newcommand{\nchunk}[1][]{%
\ifthenelse{\equal{#1}{}}{S}{S(#1)}%
}

\newcommand{\suppress}[1]{}

\title{The benefit of a $1$-bit jump-start, and the necessity of stochastic encoding, in jamming channels}
\author{Bikash Kumar Dey, Sidharth Jaggi, Michael Langberg, Anand D. Sarwate}
\date{\today}

\begin{document}

\maketitle
\begin{abstract}
We consider the problem of communicating a message $m$ in the presence of a malicious jamming adversary (Calvin), who can erase an arbitrary set of up to $pn$ bits, out of $n$ transmitted bits $\bX = (x_1,\ldots,x_n)$. The capacity of such a channel when Calvin is {\it exactly causal}, {\it i.e.} Calvin's decision of whether or not to erase bit $x_i$ depends on his observations $(x_1,\ldots,x_i)$ was recently characterized~\cite{bassily_causal_2014,chen_characterization_2015} to be $1-2p$. In this work we show two (perhaps) surprising phenomena. Firstly, we demonstrate via a novel code construction that if Calvin is {\it delayed} by even a single bit, {\it i.e.} Calvin's decision of whether or not to erase bit $x_i$ depends only on $(x_1,\ldots,x_{i-1})$ (and is independent of the ``current bit'' $x_i$) then the capacity increases to $1-p$ when the encoder is allowed to be stochastic. Secondly, we show via a novel jamming strategy for Calvin that, in the single-bit-delay setting, if the encoding is deterministic ({\it i.e.} the transmitted codeword $\bX$ is a deterministic function of the message $m$) then no rate asymptotically larger than $1-2p$ is possible with vanishing probability of error; hence {\it stochastic encoding} (using private randomness at the encoder) is essential to achieve the capacity of $1-p$ against a one-bit-delayed Calvin.
\end{abstract}

\section{Introduction}

There are two traditional methods in information theory for modeling uncertainty in communication channels. Shannon's approach treats uncertainty in the channel as a random phenomenon and requires the probability of decoding error to vanish as the blocklength tends to infinity~\cite{shannon_mathematical_1949}. The capacity is governed by the behaviour of typical channel realizations; for example in a binary erasure channel (BEC) with erasure probability $p$, the channel will erase \textit{approximately} $p n$ symbols as $n \to \infty$.  Classical error-control coding, which we might call Hamming's approach, considers the problem of worst-case recovery. Assuming the channel erases \textit{at most} $p n$ symbols, the goal is to design codes that can exactly recover the transmitted message.

One way to view the differences between these two models is to anthropomorphize the channel and assume it is being controlled by an adversary (whom we call Calvin), who wishes to foil the communication between the transmitter and receiver (hereafter referred to as Alice and Bob).  By restricting the information available to Calvin we can recover models for communication in these two regimes. This information could be about the transmitted message or the codeword itself.
For example, a BEC could be modeled by an \textit{oblivious} adversary who knows neither the message nor the codebook used by the transmitter and receiver, and is restricted to erase no more than $p n$ symbols as $n \to \infty$. In the BEC we allow for some probability of error (average or maximum over messages) that tends to $0$ as $n \to \infty$.
The Hamming approach is more pessimistic: Calvin knows the transmitted message, codeword, and codebook, and can adversarially choose up to $p n$ positions to erase to create uncertainty at the decoder. A good code in the Hamming sense protects against all such erasure attacks and guarantees zero error subject to the adversary's constraint.

The advantage of this (perhaps paranoid) adversarial modeling is that it reveals a plethora of intermediate models between the Shannon and Hamming models that can potentially shed light on the difference between average and worst case analysis. An arbitrarily varying channel (AVC)~\cite{blackwell_capacities_1960} has a time-varying state (e.g. the presence/absence of an erasure) that can be chosen by Calvin. In AVC models, distinctions between error criteria (maximum or average) and the presence of common randomness shared by Alice and Bob become important~\cite{lapidoth_reliable_1998}. Sometimes the AVC capacity displays a dichotomy: if Calvin can simulate sending a legitimate message, then Bob may not be able to decode correctly. Such AVCs are called \textit{symmetrizable}, and the capacity is $0$ in this case~\cite{csiszar_capacity_1988}.

In this paper we find a new dichotomy when Calvin can observe the transmitted codeword subject to some~\textit{delay}. That is, at time $i$, Calvin has knowledge of the transmitted codeword up to time $i - \delay$. In particular, we study the case $\delay = 1$ for a model in which Calvin can erase at most $p n$ of the transmitted bits. If the encoder and decoder share common randomness, then prior work shows that the capacity in this model is $1 - p$, the same as the BEC capacity~\cite{cover_elements_2012}. 

We show that in our model of study with $\delay = 1$ the capacity  is $1 - p$. Our coding scheme uses randomness at the encoder. Specifically, for any rate below $1 - p$, the maximum probability of error (over the encoder randomness) goes to $0$ as $n \to \infty$. This result may come as a surprise, as the capacity for $\delay = 0$ is strictly lower and equals $1 - 2p$~\cite{bassily_causal_2014,chen_characterization_2015}. Moreover, we show encoder  randomness is essential by proving that any deterministic coding scheme will have capacity at most $1-2p$ for $\Delta=1$. 
In contrast, we show in Sec.~\ref{sec:stochdet} that for omniscient adversary, 
who has noncausal knowledge of the full codeword, the capacity under stochastic
encoding is the same as that under deterministic encoding.

\begin{figure}[t!]
\begin{center}
\hspace*{-5mm}\includegraphics[scale=0.4]{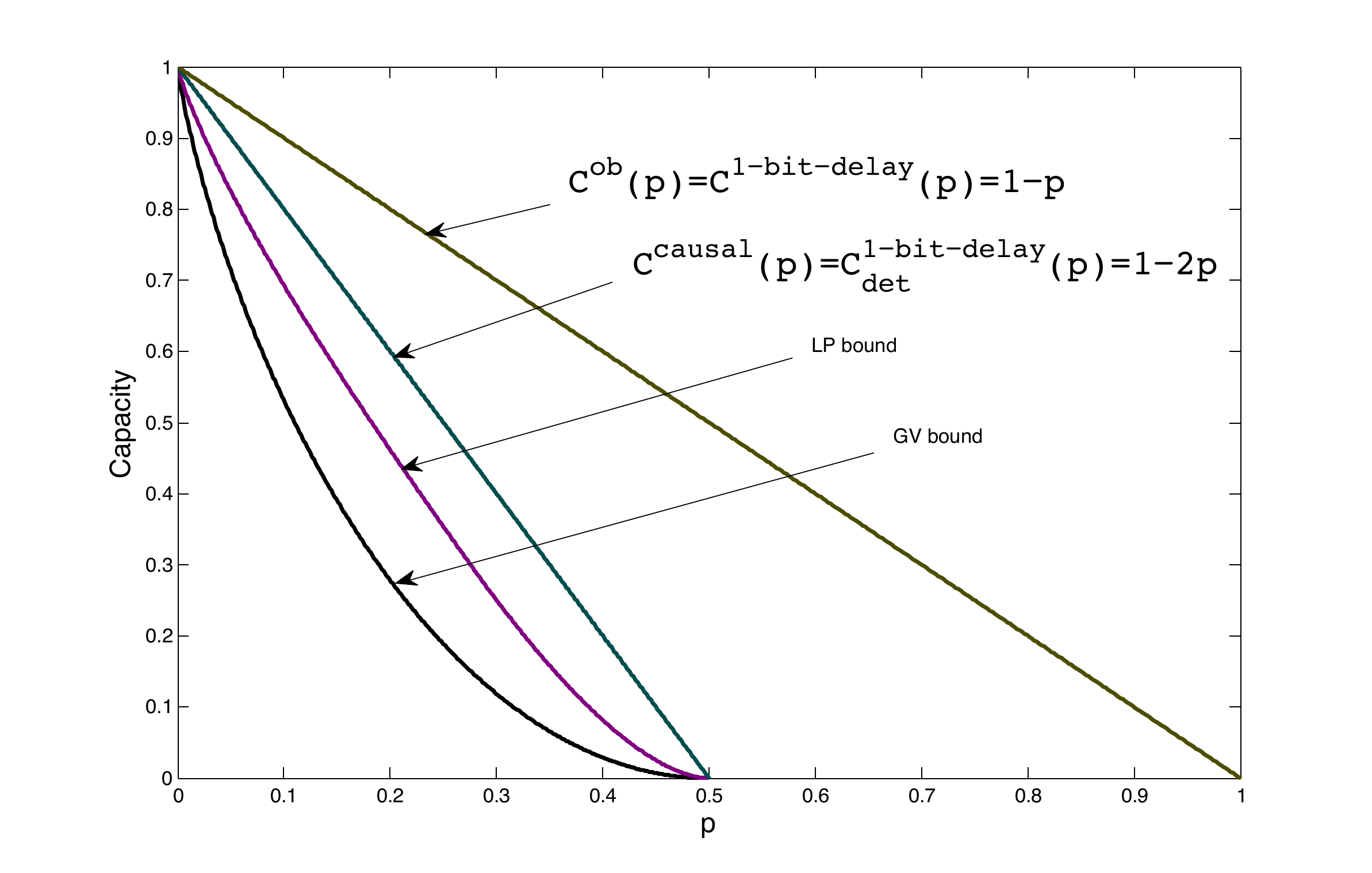}
\vspace{-6mm}
\caption{Binary adversarial erasure channels.
} \label{fig:rate-region}
\end{center}
\vspace{-6mm}
\end{figure}

\subsection{Prior work and contributions}

We focus on two aspects of communication models with adversaries: the impact of delay on the knowledge of the adversary, and the difference between deterministic and stochastic encoding. The first paper to our knowledge that examined these issues was by Ahlswede and Wolfowitz~\cite{AhlswedeW:69correlated}, who gave several equivalences between classes of AVC models and further showed that stochastic encoding alone can have some benefit over deterministic encoding. Traditional works on the AVC~\cite{lapidoth_reliable_1998} focused on the case where Calvin is oblivious ($\Delta = n$); for erasure adversaries the capacity for average error and deterministic codes is $1 - p$~\cite{csiszar_capacity_1988}. If the encoder and decoder share common randomness then the capacity is $1 - p$ even if Calvin is omniscient ($\Delta = -n$)~\cite{sarwate_rateless_2010}. For deterministic codes the the best-known achievable rate equals $1-H(p)$ via GV codes~\cite{gilbert_comparison_1952,varshamov_estimate_1957} and the best-known outer bound is given by the LP bound~\cite{mceliece_new_1977}. 

Our results show a sharp difference between $\Delta = 0$ and $\Delta = 1$. For $\Delta=0$, Bassily and Smith proved an outer bound of $1 - 2 p$ and Chen et al.~\cite{chen_characterization_2015} constructed a code with stochastic encoding that achieves $1 - 2p$. For $\Delta = -\epsilon n$ (that it, $\epsilon n$-lookahead) this code achieves $1-2p-\epsilon$, showing that sublinear lookahead cannot improve Calvin's jamming strategy. We find the capacity for $\Delta=1$ (hence the title ``one-bit delay'') is $1-p$, thereby establishing the same result for all positive $\Delta$. This demonstrates a very sharp asymmetry between the effect of lookahead and delay! The capacities of the problems in the spectrum above are plotted in Figure~\ref{fig:rate-region}.

The second issue we address is the importance of stochastic encoding (private randomization). The first paper on AVCs~\cite{blackwell_capacities_1960} considered the case $\Delta = 1$ with full common randomness, but their proof does not extend to constrained adversaries~\cite{CsiszarN:88constraints}. Most AVC results focus on the difference between common randomness and deterministic coding for oblivious~\cite{blackwell_capacities_1960,ahlswede_elimination_1978,CsiszarN:88constraints,csiszar_capacity_1988} or omniscient~\cite{langberg_private_2004,smith_scrambling_2007,sarwate_rateless_2010} adversaries. Stochastic encoding offers
few benefits in these settings for DMCs or AVCs, although it is useful in
wiretap scenarios~\cite{wyner_wire-tap_1975}. In this paper we show that for
$\Delta = 1$ deterministic codes cannot achieve rates higher than $1 - 2p$
whereas stochastic encoding can achieve a rate $1 - p$. This shows that
stochastic encoding is essential for the specific channel considered in this
work. A related (and fascinating) open question is whether the same is true for
exactly causal binary erasure channels -- the rate-optimal
codes~\cite{chen_characterization_2015} achieving $1-2p$ used stochastic
encoding, and it is unclear whether the same rate is achievable via
deterministic codes. However, for an omniscient adversary, the capacity
under stochastic encoding is argued to be the same as that under deterministic
encoding in Sec.~\ref{sec:stochdet}.

\subsection{For comparison: Large alphabet channels}

Often, analyzing ``large alphabet'' channels (where the channel input/output alphabet sizes are larger than the blocklength $n$) gives one insight about general channels, including channel models that are challenging to characterize (such as binary channels). 

For large alphabet erasure channels, the situation is somewhat different than considered in this paper. If the input alphabet $\cX$ is of size $q$ which is at least $n$, and at most a $p$-fraction of symbols may be erased, the capacity is exactly $1-p$. (This equals the capacity of the $q$-ary random erasure channel, in which the erasure probability of each symbol is $p$.) This rate is attainable regardless of the knowledge of the adversary, and computationally-efficiently attainable by (deterministic) Reed-Solomon codes; hence neither of the behaviours observed in the binary adversarial erasure channel is observed here.

On the other hand, for large alphabet symbol errors when at most $pn$ output symbols may differ from the input symbols, we may observe similar behaviour to the binary erasure case. In~\cite{dey_codes_2013}, it was demonstrated that the capacity of exactly causal channels equals $1-2p$, which is the same as the capacity if the adversary is omniscient -- hence there is no advantage of lookahead for the adversary. In both cases we can achieve these rates using computationally efficient  (and deterministic) Reed-Solomon codes. However, if the adversary is delayed, then, depending on the symbol-error model, the capacity may be higher. Two symbol-error models were considered. When symbol errors are {\it additive} (the output symbol $y_i$ equals $x_i + e_i$, where $x_i$ is the input symbol, $e_i$ is the error symbol,  at most $pn$ $e_i$s may be non-zero, and all symbols and addition are over the finite field ${\mathbf F}_q$) with a delay of even a single symbol, the capacity equals $1-p$ (thereby exhibiting a similar phase-transition in the capacity as in this paper). In contrast, with {\it overwrite} errors (the output symbol $y_i$ equals $e_i$ for at most $pn$ non-zero $e_i$s) with a delay of $dn$ symbols ($e_i$ can be a function of $x_1,\ldots,x_{i-dn}$), the capacity is $1-2p+d$ for $p<1/2$, and $0$ otherwise, hence demonstrating a less sharp transition in the throughput.

The differences in optimal rates obtainable with stochastic and deterministic encoding over large alphabets with causally-constrained adversaries has not, to the best of our knowledge, been considered in the literature and may be worthy of investigation.

\section{Channel model}

For integers $r < s$ let $[r:s]$ denote the set
$\{r,r+1,\ldots, s\}$ and let $[N]$ denote the set $[1:N]$.  For a set $S
\subseteq [n]$, let $\bar{S}$ be the complement of $S$. Let $\er$ denote the
erasure output symbol. Random variables will typically be denoted by capital letters
and vectors by boldface. For a vector $\bz = (z_1, z_2, \ldots, z_n)$ and set $S
\subseteq [n]$ we will write $\bz_{S}$ for the vector $(z_i)_{i \in S}$ with
the components ordered in increasing order of index. The Hamming weight of a 
binary vector $\bz$ is $\wt(\bz)$, and Hamming distance is $\dH$.


We first set up our channel model more generally before 
specializing to the case considered in this writeup.

Let $\cX$,$\cY$,and $\cZ$ be discrete alphabets. We consider variants on arbitrarily varying channel models, which are channels whose state $z \in \cZ$ is (partially) controlled by a malicious adversary who wishes to prevent reliable communication across the channel. The model is parameterized by a set of discrete channels $\{ W(y | x, z) : x \in \cX, y \in \cY, z \in \cZ \}$. For blocklength $n$, input $\bx \in \cX^n$, state $\bz \in \cZ^n$, and output $\by \in \bY^n$, the blocklength-$n$ extension of this channel is
	\begin{align}
	W(\mbf{y} | \mbf{x},\mbf{z}) = \prod_{i=1}^{n} W(y_t | x_t, z_t).
	\end{align}
	
An $(n, 2^{nR})$ code with randomized encoding for this channel is a pair of maps $(\renc,\ddec)$ where $\Phi : [2^{nR}] \to \cX^n$ is a randomized encoding map and $\psi : \cY \to \{0\} \cup [2^{nR}]$ is a deterministic decoding map. 
In a deterministic code, the encoder is also deterministic, and it assigns
a unique codeword to each of the $2^{nR}$ messages.

We consider channel models in which $\bz$ is chosen adversarially and with partial knowledge of the transmitted codeword. We define an adversarial strategy $\Gamma$ of delay $\delay$ to be a sequence of maps $\{\gamma_t : t \in [n]\}$, where $\gamma_t : \cX^{t-\delay} \to \cZ$ is a randomized map from $\bx_{[1:(t - \delay)]}$ to $z_t$. We allow this map to depend on the code $(\renc,\ddec)$. Alternatively, such a strategy defines a conditional probability distribution $G( z_t | \bx_{[1:(t - \delay)]}, \renc,\ddec)$ which chooses $Z_t$. The corresponds to a scenario where the adversary can observe the channel input up to delay $\delay$ and can choose the channel state based on that information and the structure of the code.  We say the strategy satisfies a cost constraint $p$ with respect to the cost function $\cost : \cZ \to \mathbb{R}^{+}$ if 
	\begin{align}
	\sum_{t=1}^{t} c(Z_t) \le p n. 
	\end{align}
Let $\jamset(p,\renc,\ddec)$ be the set of of strategies that satisfies the cost constraint.

The probability of error for this code on message $m \in [2^{nR}]$ with adversarial strategy $\Gamma$ is
	\begin{align}
	\Perr(m,\Gamma) = \sum_{\by : \ddec(\by) \neq m} \sum_{\bx\in\cX^n} \sum_{\bz \in \cZ^n} 
		W(\by|\bx,\bz) \P(\renc(m) = \bx) \prod_{t=1}^{n} G(z_t | \bx_{[1:(t - \delay)]}, \renc,\ddec).
	\end{align}
The maximum probability of error is 
	\begin{align}
	\Perr = \max_{\Gamma \in \jamset(p,\renc,\ddec)} \max_{m \in [2^{nR}]} \Perr(m,\Gamma).
	\end{align}
Note that these probabilities are over the encoder randomness in $\renc$, potential randomness in the adversary strategy $\Gamma$, and possible randomness
in the channel.
We say a rate $R$ is achievable in this model if there exists a sequence of $(n, 2^{\lfloor nR \rfloor})$ codes such that $\Perr \to 0$ as $n \to \infty$. The capacity is the supremum of the set of achievable rates.

Here we take $\cX = \{0,1\}$, $\cZ = \{0,1\}$, and $\cY = \{0,1,\perp\}$. The channel model is given by $y_t = x_t$ if $z_t = 0$ and $y_t = \perp$ if $z_t = 1$. The cost function is $c(z) = z$ and $\delay = 1$. This corresponds to a binary-input channel in which the adversary can observe all past inputs and can erase up to $p n$ of the bits. Under a larger delay the capacity is $1 - p$. However, for delay $0$ the capacity is $1 - 2p$. In the remainder of the paper we will show that with stochastic encoding the capacity is $1 - p$ and with deterministic encoding the capacity is at most $1 - 2p$.

Our main results take the form of two theorems. Theorem \ref{th:stochastic} is an achievability result: it says that the stochastic encoding can achieve rate $1 - p$ against a bit-erasing adversary who can erase up to $p n$ bits and is subject to delay $\delay = 1$. 

\begin{theorem}
\label{th:stochastic}
The capacity of a binary channel with a bit-erasing adversary who can
erase up to $p$ fraction of a codeword based on
causal $1$-bit-delayed observation is $1-p$.
\end{theorem}

The next theorem contrasts the above result to say that if the transmitter
is restricted to using deterministic codes, then the capacity is at most $1-2p$.
Therefore stochastic encoding is crucial to take advantage of
the adversary's delayed observation.

\begin{theorem}
\label{th:deterministic}
The capacity of a binary channel under deterministic encoding,
with a bit-erasing adversary who can
erase upto $p$ fraction of a codeword based on
causal $1$-bit-delayed observation is at most $1-2p$.
\end{theorem}

\section{Analysis for Stochastic Encoding: Proof of Theorem~\ref{th:stochastic}}
\label{sec:lower}

We consider coding for an online adversarial channel with binary inputs in
which the adversary observes the channel input subject to unit delay and can
erase a fraction $p$ of the bits. Under a larger delay the capacity is $1 - p$
bits. However, for delay $0$ the capacity is $1 - 2p$ bits. 

\subsection{Code construction, encoding, and decoding}

Given a parameter $\capgap > 0$, rate $R = 1 - p - \capgap$ and blocklength $n$, let $M = \lfloor 2^{nR} \rfloor$ be the number of messages.

\subsubsection{Random code construction}
Our code construction relies on the following parameter settings:
	\begin{align}
	K &= (1/4) \log_2 n 
		\label{eq:numprobs} \\
	q_k &= 2^{k-1}n^{-1/2}, \qquad k \in [K] 
		\label{eq:fuzzprobs} 
	\end{align}
	
        \begin{enumerate}
        \item For each message $m \in [M]$ there is a \textit{base codeword} $\bu(m)$, selected uniformly at random from $\{0,1\}^n$.
        \item For each message $m \in [M]$ the encoder has a partition $\{
\chunk{m}{k} : k \in [K] \}$ of the set $[n]$ so that for each $k \in [K]$ the
set $\chunk{m}{k}$ is a set of indices of the codeword. 
We generate the partitions $\{\chunk{m}{k}\}_k$ for each $m$ by binning the indices in $[n]$ into $K$ bins independently and uniformly at random.
        \item In addition the
encoder maintains a set $\mc{Q} = \{ q_k : k \in [K] \}$ of probabilities, where $K$ and $q_k$ are given by \eqref{eq:numprobs} and \eqref{eq:fuzzprobs}.
        \end{enumerate}
	
\subsubsection{Encoding} 

The encoding is randomized.  To encode the message $m \in [M]$
the encoder transmits $\bX = \bu(m) \oplus \bZ$, where $\bZ = (Z_1, Z_2,
\ldots, Z_n)$ with $Z_i \sim \bern( q_k )$ if $i \in \chunk{m}{k}$. That is,
for each $k\in K$, the encoder adds a $\bern(q_k)$ noise to the components
indexed by $\chunk{m}{k}$. Our encoder is depicted in Figure~\ref{fig:encoder}.

\begin{figure}
\begin{center}
\includegraphics[width=0.8\textwidth]{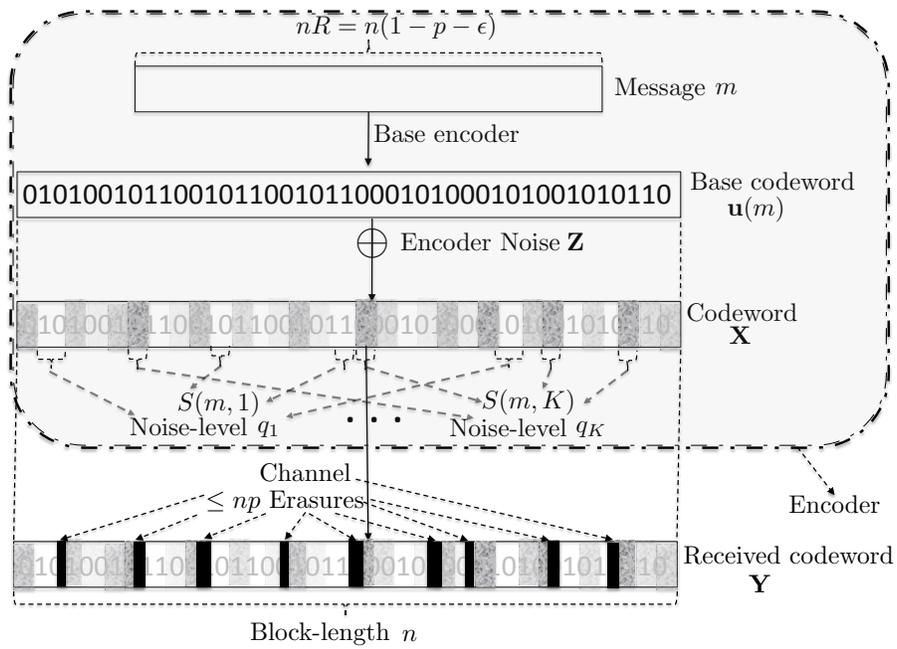}
\caption{Encoder: the message $m$ is encoded first to ${\bf u}(m)$ and then noise ${\bf Z}$ is added according to the subsets $\{S(m,k)\}$ and corresponding noise level probabilities $\{q_k\}$. The subsets are represented by different shades of grey. The resulting codeword X is corrupted by Calvin who may use at most $pn$ erasures. The received word is $Y$.}
\label{fig:encoder}
\end{center}
\end{figure}

\subsubsection{Decoding}

\begin{enumerate}
\item Given the received codeword $\bY$, the decoder first finds the smallest
index $\decT$ such that the first $\decT$ positions of $\bY$ contain
$(R+\decslack)n$ unerased bits: \begin{align}
        \decT = \min \{ t : |\{ i: i\leq t, \bY_i \neq \perp \}| \ge (R+\decslack)n \}.\label{eq:eta2}
        \end{align}

\item {\bf (List decoding):} The decoder then constructs a list $\mc{L}$ based on the prefix $\bY_1^{\decT}$. More specifically, message $m$ is put in the list if
        \begin{align}
        \left| \left\{ i \in [1:\decT] : u_i(m) \ne Y_i, Y_i \neq \er \right\} \right| < n^{3/4}. \label{eq:list-size}
        \end{align}
That is, all codewords which are sufficiently close in Hamming distance (on the unerased bits) are put in the list.

\item {\bf (List disambiguation):} The decoder then turns to the suffix $\bY_{[(\decT+1):n]}$. For a tuple $(m_1,m_2,k_1,k_2)$ define the set of unerased bits that are in the $k_1$-th part of $m_1$ and the $k_2$-th part of $m_2$:
        \begin{align}
        V_{m_1,m_2,k_1,k_2} = & \left\{ i \in [(\decT+1):n] \cap \chunk{m_1}{k_1} \cap \chunk{m_2}{k_2} : Y_i \in \{0,1\} \right\}.
        \end{align}
For each pair $(m_1,m_2) \in \mc{L} \times \mc{L}$, the decoder first checks to see if there exists a $(k_1,k_2)$ such that $k_1 \ne k_2$ and
        \begin{align}
        \left| V_{m_1,m_2,k_1,k_2} \right| \ge \frac{\capgap n}{4(K^2-K)}. \label{eq:decoder}
        \end{align}
\noindent If no such pair $(k_1,k_2)$, $k_1 \ne k_2$, exists then the decoder
declares a decoding error. If such a pair exists the decoder takes the first
such pair (lexicographically ordered) over all $K^2-K$ such pairs, which we
denote by $V_{m_1,m_2}$.

We adopt a simplified maximum likelihood decoding rule. Partition the set of indices into positions where $\bu(m_1)$ and $\bu(m_2)$ agree or disagree:
        \begin{align}
        V_0 &= \{i \in V_{m_1,m_2} : u_i(m_1) = u_i(m_2) \} \\
        V_1 &= \{i \in V_{m_1,m_2} : u_i(m_1) \ne u_i(m_2) \}.
        \end{align}
Set $V$ to be the larger of the two sets so that $|V| \ge | V_{m_1,m_2}
|/2$. We apply the maximum likely decoder to $V$. Let $\alpha( m ) = \dH( \bY_{V}, \bu_{V}(m) )$. We say $m_1$ beats $m_2$ if
        \begin{align}
        \frac{ q_{k_1}^{\alpha(m_1)} (1 - q_{k_1})^{|V| - \alpha(m_1)} 
                }{ 
                q_{k_2}^{\alpha(m_2)} (1 - q_{k_2})^{|V| - \alpha(m_2)}
                } 
                > 1,
          \label{eq:mlrule}
        \end{align}
otherwise we say $m_2$ beats $m_1$.
	
\item If there exists a message $\hat{m}$ in the list $\mc{L}$ that beats all other elements of the list (a ``Condorcet winner'') then output that message $\hat{m}$, else it declares an error.
\end{enumerate}

\subsection{Analysis}

In the analysis we follow the usual recipe: we show that for sufficiently large
$n$, with high probability, a randomly constructed code 
will have $\Perr$ that vanishes as $n \to \infty$, thereby showing that
such a code exists.  Recall that in our code construction, we choose 
the `pure codewords' $\bu(m)$
independently and uniformly at random from $\{0,1\}^n$ (i.e.~they are
i.i.d.~$\bern(1/2))$; and we generate the partitions $\{\chunk{m}{k}\}_k$ 
for each $m$ by binning the indices in $[n]$ into $K$ bins uniformly at random.
The codebook consists of both the `pure codewords' as well as the
partitions for each message.

Let the random variable representing this codebook be denoted by $\cC$.  We will prove that the codebook has nice properties with a probability that is super-exponentially close to $1$.

\begin{figure}
\begin{center}
\includegraphics[width=0.8\textwidth]{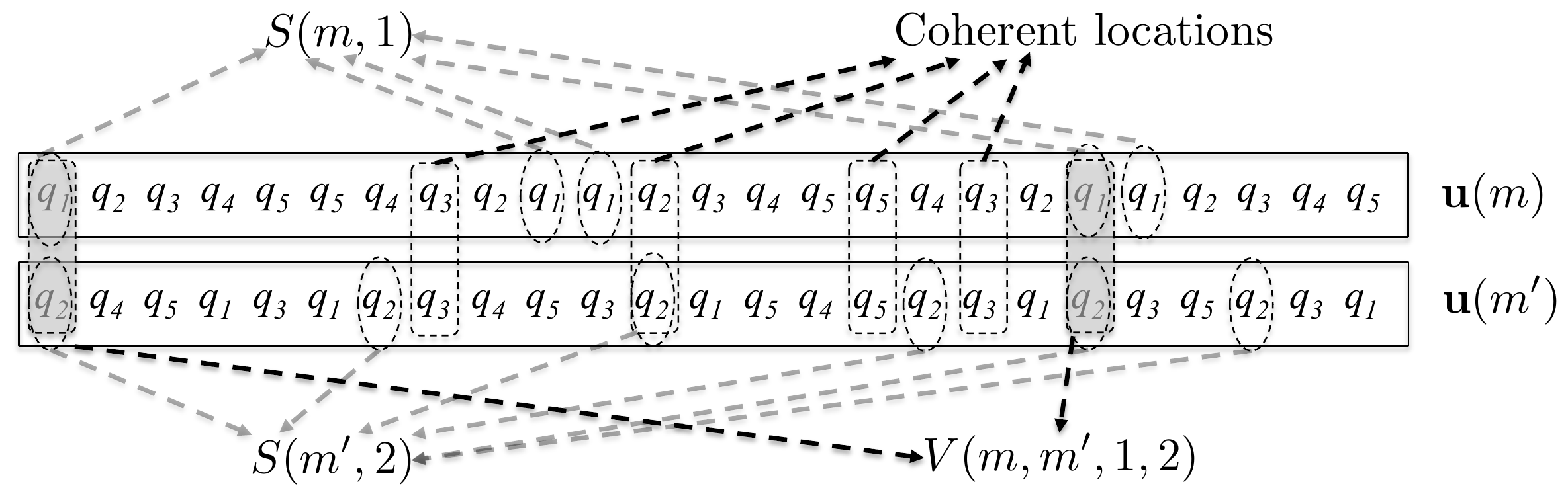}
\caption{An example demonstrating coherence over ${\cal T}$ for two base codewords $\bu(m)$ and $\bu(m')$. Let ${\cal T} = [25]$, {\it i.e.}, it comprises of the first $25$ locations of each codeword, and let there be $K=5$ noise-levels $q_1,\ldots,q_5$ for each codeword, in the sets of locations $S(m,1),\ldots,S(m,5)$ and $S(m',1),\ldots,S(m',5)$ respectively. The expected size of each $V(m,m',k,k')$ is therefore $|{\cal T}|/K^2 = 1$. It can be verified that the largest size of $V(m,m',k,k)$ is $2$ (only for $k = 3$ -- all other sets of size at least $2$ have $k\neq k'$), and hence $\bu(m)$ and $\bu(m')$ are at most $2$-coherent over ${\cal T}$. Further, there are exactly $4$ locations in which  $\bu(m)$ and $\bu(m')$ are coherent (the $8$th, $12$th, $16$th and $18$th locations, as highlighted in this figure), so the remaining $21$ decoherent locations are potentially usable by the decoder Bob, to disambiguate between $m$ and $m'$. In this example, $V(m,m',1,2)$ comprising of the two locations $\{1,20\}$ is a possible choice for the disambiguation set, being of ``reasonable size'', and being the lexicographically first set with $k \neq k'$.} 
\label{fig:coherence}
\end{center}
\end{figure}

Fix any $\capgap > 0$ and recall $R = 1-p-\capgap$. We prove a sequence of lemmas to prove we can achieve rate $R$.

\begin{lemma} \label{lem:zweight}
With probability at least $1-\exp \left (-\frac{\sqrt{n}}{2} \right )$ over encoder's random noise $\bZ$, the Hamming weight of $\bZ$ is at most $n^{3/4}$.
\label{lem:noise}
\end{lemma}

\begin{proof}
Since each $\P(Z_j = 1) \le q_K$ for all $j$, the probability is upper bounded by the probability that $n$ i.i.d. variables $A_j \sim \bern(q_K)$ have Hamming weight greater than $n^{3/4}$. The expected weight of $\mbf{A}$ is $q_Kn=(2^{(\log(n)/4)-1}n^{-1/2})n = \frac{n^{3/4}}{2}$. Therefore by Hoeffding's inequality,
	\begin{align}
	\P\left( \sum_{j} Z_j > n^{3/4} \right) 
		&\le \exp \left (-\frac{\sqrt{n}}{2} \right ).
	\end{align}
\end{proof}

\begin{lemma}
The length $\decT$ given in \eqref{eq:eta2} of the prefix $\bY_1^\decT$ is at most $(1-\capgap/2)n$ and the suffix $\bY_{\decT+1}^n$ has at least $n\capgap/2$ unerased bits.
\label{lem:prefix}
\end{lemma}

\begin{proof}
Since $R = 1-p-\capgap$ and the adversary can erase at most $pn$ locations, the number of unerased bits in $\bY_1^{(1-\capgap/2)n}$ is at least $(1-p-\capgap/2)n = n(R+\capgap/2)$, as required by the definition of $\decT$ in \eqref{eq:eta2}. 
Let $\lambda$ be the number of unerased bits in $\bY_1^\decT$.
By our definitions it holds that $\lambda = (1-p-\capgap/2)n$.
Thus the number of erased bits in $\bY_1^\decT$ is $\decT  - \lambda$.
Implying at most $pn-\decT  + \lambda$ erased bits in $\bY_{\decT+1}^n$, which finally implies at least
	\begin{align}
	(n-\decT) - (pn - \decT  +\lambda) = n-pn-(1-p-\capgap/2)n = \capgap n/2
	\end{align} 
unerased bits in $\bY_{\decT+1}^n$.
\end{proof}

	
%
%


We now define a few useful properties of our random code. The decoder will have difficulty resolving the difference between codewords if they share very similar partitions $\{S(m,k)\}$.  For two messages $(m,m')$, set $\cT \subseteq [n]$, the expected number of common locations over $\cC$ is
	\begin{align}
	\E_{\cC}\left[ \sum_{k=1}^{K} |\chunk{m}{k}\cap \chunk{m'}{k} \cap \cT| \right] 
	= \sum_{k=1}^K \frac{|\cT|}{K^2} 
	= \frac{|\cT|}{K}.
	\end{align}
For a $\cT \subset [n]$, call a pair of base codewords $(\bu(m), \bu(m'))$ \textit{$\setcoh$-coherent over $\cT$} if 
	\begin{align}
	\sum_{k=1}^K |\chunk{m}{k}\cap \chunk{m'}{k} \cap \cT| \le \frac{|\cT|}{K}(1+\setcoh)
	\end{align}
That is, the number of locations in $\cT$ in which both $\bu(m)$ and $\bu(m')$ have the same noise levels is at most an $(1+\setcoh)$ factor greater than the expected number of such locations.
We call a codebook $\cC$ is \textit{$(\setcoh,\codecoh)$-coherent} if for each pair of messages $(m,m')$ and each $\cT$ of size at least $\codecoh n$, the pair of base codewords $(\bu(m), \bu(m'))$ are $\setcoh$-coherent over $\cT$.

Let $\cTp$ be an ordered subset of $[n]$, and denote its $i$th entry by $(\cTp)_i$.
Define the {\it restriction of a base codeword $\bu(m)$ to $\cTp$}, $\bu_\cTp(m)$, as the length-$|\cTp|$ binary vector whose $i$-th entry equals the $(\cTp)_i$th entry of $\bu(m)$. Define the {\it restriction of a codebook $\cC$ to $\cTp$}, denoted $\cC_{\cTp}$, is analogously defined as the codebook (with possible repetitions) generated by restricting each base codeword $\bu(m) \in \cC$ to $\cTp$.
We call codebook $\cC$ {\it $(w_\uner,w_e,s)$-list-decodable} if for each set $\cTp \subset [n]$ of size at least $w_\uner$ (of unerased bits), the restricted codebook $\cC_{\cTp}$ is ``decodable against weight $w_e$ errors to a list of size at most $s$''. More precisely, for the (unrestricted) codebook $\cC$, for any set $\cTp \subset [n]$ of size at most $w_\uner$, any Hamming ball in $\{0,1\}^{w_\uner}$ of radius at most $w_e$ contains less than $s$ codewords restricted to $\cTp$.

\begin{lemma}
\label{lem:props}
For any sufficiently small $\capgap >0$, there exists sufficiently large $N_\capgap$, such that for all $n>N_\capgap$, with probability at least 
	\begin{align}
	1-2^{-\frac{\capgap^2 }{4}n\log\log(n)}
	\end{align}
over the design of codebook $\cC$, the following two properties hold: 
\begin{enumerate}
\item \label{lem:coh} The codebook $\cC$ is at most $(K/2-1,\capgap/2)$-coherent.
\item \label{lem:list} The codebook $\cC$ is $(n(1-p-\capgap/2),n^{3/4},(\log\log(n))\capgap/2)$-list-decodable.
\end{enumerate}
\end{lemma}

\begin{proof}
We first prove that with high probability $\cC$ is at most $(K/2-1,\capgap/2)$-coherent. Since $\setcoh = (K/2-1)$, we must show that for any pair of base codewords $(\bu(m), \bu(m'))$
	\begin{align}
	\sum_{k=1}^K |\chunk{m}{k}\cap \chunk{m'}{k} \cap \cT| \le \frac{|\cT|}{2}.
	\end{align}
Recall that the sets $\{\chunk{m}{k}\}_{k=1}^K$ partition $[n]$.  We will calculate the probability over the randomly selected partitions $\{\chunk{m}{k}\}$.

Fix any set $\cT \subset [n]$ with size at least $\capgap n/2$. The probability that the code construction generates $\frac{|\cT|}{2}$ or more positions in which $m$ and $m'$ select the same $q_k$ can be written as follows:
\begin{align}
\sum_{i=\frac{|\cT|}{2}}^n \binom{n}{i} \left(\frac{1}{K} \right )^i \left(1-\frac{1}{K} \right )^{n-i}  
&< \sum_{i=\frac{|\cT|}{2}}^n \binom{n}{i} \left(\frac{1}{K} \right )^i \\ 
&<  \sum_{i=\frac{|\cT|}{2}}^n 2^n \left(\frac{1}{K} \right )^i\\ 
&<  n2^n \left(\frac{1}{K} \right )^{\frac{|\cT|}{2}}\\ 
&< n2^n \left(\frac{4}{\log(n)} \right )^{\frac{\capgap n}{4}}\\
&= 2^{-\frac{\capgap }{4}n\log\log(n)+ (1+\capgap/2)n+\log(n)},
\end{align}
where the last inequality follows from the setting of $K$ as $\log(n)/4$ and the size of $\cT$ being at least $n \capgap/2$. Taking a union bound over all pairs of base codewords (there are strictly less than $2^{2n}$ such pairs, since the rate of the code is less than $1$) and all possible sets $\cT$ (there are strictly less than $2^n$ such sets) shows that the probability that a code is {\it not} at most $(K/2-1,\capgap/2)$-coherent is at most $2^{-\frac{\capgap }{4}n\log\log(n)+ (4+\capgap/2)n+\log(n)}$.

We now prove that with high probability $\cC$ is appropriately list-decodable. This is broadly similar to classical derivations of list-decoding bounds, but due to the specific combination of error/erasure decoding required in this proof (with asymptotically vanishing fraction of errors but constant fraction of erasures) we re-derive a proof here. Since each base codeword $\bu(m)$ in the codebook $\cC$ is generated uniformly at random from $\{0,1\}^n$, the same is true for codewords restricted to $\cTp$ (for any $\cTp$). Therefore for all sufficiently large $n$, the probability that a codeword in $\cC_\cTp$ falls in any fixed Hamming ball of radius $n^{3/4}$ in  $\{0,1\}^{n(1-p-\capgap/2)}$ is  
	\begin{align}
	\frac{ \binom{n(1-p-\capgap/2)}{n^{3/4}} }{ 2^{n(1-p-\capgap/2)} }
	=
	2^{-n(1-p-\capgap/2)  + n^{3/4}\log(n^{1/4}) + \bigO(n^{1/4})} < 2^{-n(1-p-2\capgap/3)}
	\end{align}
where the equality follows from Stirling's approximation.

Let $\nu = 2^{-n(1-p-2\capgap/3)}$. The probability (over the design of codebook $\cC$) then that the Hamming ball contains at least $(\log\log(n))\cohbound$ codewords restricted to $\cTp$ is at most\footnote{Note that the expected number of codewords in the Hamming ball is no more than $q2^{nR}$, which equals $2^{-n\capgap/3}$.},
	\begin{align}
	\sum_{i=(\log\log(n))\capgap/2}^{{2^{nR}}} \binom{2^{nR}}{i} q^i (1-q )^{2^{nR}-i}  
	&<  \sum_{i=(\log\log(n))\capgap/2}^{{2^{nR}}} \binom{2^{nR}}{i} \nu^i \\ 
	& <  \sum_{i=(\log\log(n))\capgap/2}^{{2^{nR}}} 2^{nRi} \nu^i \\ 
	&=  \sum_{i=(\log\log(n))\capgap/2}^{{2^{nR}}} 2^{n(1-p-\capgap)i} \left (2^{-n (1-p-2\capgap/3)}	\right )^i \\ 
	&< 2^{nR} 2^{-\frac{\capgap^2}{3}n\log\log(n)} \\
	&< 2^{-\frac{\capgap^2}{3}n\log\log(n) + n}.
	\end{align}
Taking a union bound over all $2^{n(1-p-\capgap/2)}<2^n$ Hamming balls and all $\binom{n}{n(1-p-\capgap/2)}<2^n$ sets $\cTp$ implies that the probability (over design of $\cC$) that there exists a set $\cTp$ for which there exists a Hamming ball with at least $(\log\log(n))\capgap/2$ codewords restricted to $\cTp$ is at most $2^{-\frac{\capgap^2}{3}n\log\log(n) + 3n}$.

Hence the probability that at least one of the two properties (approximate decoherence, and list-decodability) required do not hold for the codebook $\cC$ is at most $2^{-\frac{\capgap }{4}n\log\log(n)+ (4+\capgap/2)n+\log(n)}+2^{-\frac{\capgap^2}{3}n\log\log(n) + 3n}$, which is at most $2^{-\frac{\capgap^2 }{4}n\log\log(n)}$ for all sufficiently small $\capgap$ and sufficiently large $n$.
\end{proof}

\begin{lemma}\label{lem:code_prop}
With probability at least 
	\begin{align}
	1-2^{-\frac{\capgap^2 }{4}n\log\log(n)} 
	\label{eq:prob:code_prop}
	\end{align}
over the design of codebook $\cC$, 
for any adversarial erasure pattern $\be$,
	\begin{enumerate}
	\item For every message $m$, the size of list $\mc{L}$ in Equation \eqref{eq:list-size} is at most $\log\log(n)\capgap/2$ if $\wt(\bZ) \le n^{3/4}$, and
	\item For every pair of messages $(m_1,m_2)$, there exists a pair $(k_1,k_2)$, with $k_1\neq k_2$, satisfying Equation (\ref{eq:decoder}).
	\end{enumerate}
\end{lemma}

%
\begin{proof}
\begin{enumerate}
\item 
By Lemma~\ref{lem:prefix}, the prefix-length $\decT$ is at most $(1-\capgap/2)n$. Hence using part \ref{lem:list} of Lemma~\ref{lem:props} gives us the required bound on list-decodability.\footnote{In fact Lemma~\ref{lem:props}.\ref{lem:list} provides stronger guarantees than are required in this proof. For one, it shows list-decodability for {\it any} $\cT$ of appropriate size, whereas the decoder only ever decodes using $Y_1^\decT$. Furthermore, part \ref{lem:list} of Lemma~\ref{lem:props} guarantees that \textit{any} Hamming ball of appropriate radius does not correspond to too many messages, rather than just those Hamming balls centred at sub-vectors of $Y_1^\decT$. Neither of these relaxations asymptotically worsens the parameters obtainable in this proof, but they have the advantage of significantly simplifying presentation.}
\item We set $\cT$ to the indices corresponding to unerased bits in $[\decT+1:n]$ (which is of size at least $n\capgap/2$ by Lemma~\ref{lem:prefix}). Part \ref{lem:props} of Lemma~\ref{lem:coh} shows that with probability at least $1-2^{-\frac{\capgap^2 }{4}n\log\log(n)}$, the size of the set $\bigcup_{k=1}^K \chunk{m_1}{k}\cap \chunk{m_2}{k} \cap \cT$ is at most $\frac{|\cT|}{2}$. Hence for any $(m_1,m_2)$ the set 
	\begin{align}
	\left| \bigcup_{(k,k')\in[K]\times[K], k\neq k'} \chunk{m_1}{k}\cap \chunk{m_2}{k'} \cap \cT \right| 
		\ge
		\frac{ |\cT| }{ 2 }
		\ge \frac{ n\capgap }{ 4 }.
	\end{align}
But there are at most $K^2-K$ values for the pair $(k,k') \in [K] \times [K]$
such that $k \neq k'$. Hence for at least one such pair, the size of
$\chunk{m_1}{k}\cap \chunk{m_2}{k'} \cap \cT$ is at least $\frac{\capgap
n}{4(K^2-K)}$, as required by \eqref{eq:decoder}.  \end{enumerate}
\end{proof}



\begin{lemma} \label{lem:mldecoder}
For a code satisfying the two conditions of Lemma~\ref{lem:code_prop},
there exists a constant $c$ such that for sufficiently large $n$, with probability at least $1 - \exp\left( - \capgap \beta n^{1/2}{/ 2\log^2 n } \right)$ over the encoder noise $\bZ$,
the decoder outputs the transmitted message $m$.
\end{lemma}


\begin{proof}
Suppose $m$ was transmitted and consider the test in the decoding rule for $m_1 = m$ and $m_2 = m' \ne m$.  Let $q = q_{k_1}$ and $q' = q_{k_2}$. Let $\zeta$ denote the fraction of 1's in $\bZ_V$ (i.e. its type).  If $V = V_0$ the decoder is a the maximum likelihood detector with $|V|$ observations between hypotheses $Z_i \sim \bern(q)$ and $Z_i \sim \bern(q')$. Message $m$ beats $m'$ if \cite[(11.194)]{cover_elements_2012}:
	\begin{align}
	D( \zeta \| q' ) - D( \zeta \| q) 
	&= \zeta \log \frac{q}{q'} + (1 - \zeta) \log \frac{1 - q}{1 - q'} \\
	&> 0.
	\end{align}
If $V = V_1$ it is between $Z_i \sim \bern(q)$ and $Z_i \sim \bern(1 - q')$, so $m$ beats $m'$ if
	\begin{align}
	D( \zeta \| 1 - q' ) - D( \zeta \| q) 
	&= \zeta \log \frac{q}{1-q'} + (1 - \zeta) \log \frac{1 - q}{q'} \\
	&> 0.
	\end{align}
In both cases we can solve for the $\zeta^*$ at the threshold (where the left side equals $0$). By Sanov's Theorem~\cite[Theorem 11.4.1]{cover_elements_2012}, the probability of error is
	\begin{align}
	\P( \text{$m'$ beats $m$} ) &\le |V| \exp\left( - |V| D(\zeta^* \| q) \right).
	\end{align}
Thus we must lower bound the divergence in both cases. Since $q,q' \ll 1/2$ it is clear that the case $V = V_0$ will have a smaller upper bound, so we focus on that case. For $V = V_0$ the error is largest when the hypotheses are closest, so $|k_1 - k_2| = 1$. 

We first prove a useful lower bound on divergences. Using Taylor expansion, for $r \in (0,1)$ and $\lambda > 0$ such that $\lambda r < 1$,
	\begin{align}
	D( \lambda r \| r ) &= \lambda r \ln \lambda + (1 - \lambda r) \ln \frac{ 1 - \lambda r }{ 1 - r } \\
	&= \lambda r \ln \lambda + (1 - \lambda r)
		 \left( \sum_{j=1}^{\infty} \frac{r^j}{j} - \sum_{j=1}^{\infty} \frac{\lambda^j r^j}{j} \right) \\
	&= \lambda r \ln \lambda 
		+ (1 - \lambda r) r
		+ (1 - \lambda r) \sum_{j=2}^{\infty} \frac{r^j}{j}
		- \lambda r
		+ \sum_{j=2}^{\infty} \left( \frac{1}{j-1} - \frac{1}{j} \right) \lambda^j r^j \\
	&> r ( \lambda \ln \lambda - \lambda + 1) - \lambda r^2.
	\end{align}
Now, $\lambda \ln \lambda - \lambda + 1 = 0$ at $\lambda = 1$ and
	\begin{align}
	\frac{d}{d \lambda} ( \lambda \ln \lambda - \lambda + 1) = \ln \lambda
	\end{align}
so the coefficient of $r$ is strictly positive for all $\lambda \ne 0,1$. Thus
for sufficiently small $r$, for any $\lambda > 0, \,\lambda \neq 1$ 
there exists a $\beta > 0$ such that $D( \lambda r \| r ) \ge \beta r$.

Now we will apply this to our divergence for the threshold. We either have $q' = q/2 < \zeta^* < q$ or $q < \zeta^* < q' = 2q$, which means $r < \zeta^* < 2r$ for $r = q/2$ or $q$, and $D( \zeta^* \| r) = D( \zeta^* \| 2r)$. Therefore either $|\zeta^* - r| > r/2$ or $| 2r - \zeta^*| > r/2$, which means that $D( \zeta^* \| r) > D( 3r/2 \| r)$ or $D( \zeta^* \| 2r ) > D( 3r/2 \| 2r)$. In either case, the previous argument shows that there exists a $\beta > 0$ such that $D( \zeta^* \| r ) \ge \beta r$. Therefore
	\begin{align}
	D(\zeta^* \| q) \ge \beta q \ge \beta n^{-1/2}.
	\end{align}



Let $\mc{E}^c$ be the event that the conditions in Lemmas \ref{lem:zweight}, \ref{lem:prefix},  \ref{lem:props}, and \ref{lem:code_prop} hold. Taking a union bound over all messages $m'$ in the list, we use the fact that with
\begin{align}
        \P( \text{any $m' \ne m$ beats $m$} | \mc{E}^C=c) &\le |\mc{L}| |V| \exp\left( - |V| D( \zeta^* \| q) \right) \\
        &= (\log \log n) \frac{\capgap n}{\log^2 n} \exp\left( - \capgap \beta n^{1/2}{/ \log^2 n } \right)\\
        &\leq \exp\left( - \capgap \beta n^{1/2}{/ 2\log^2 n } \right).
        \end{align}
\end{proof}

These lemmas together imply Theorem~\ref{th:stochastic} as argued below.


{\em Proof of Theorem~\ref{th:stochastic}:}
The converse follows by considering an adversary that erases the first $pn$ bits.
Fix $\capgap > 0$ and set $R = 1 - p - \capgap$. Fix any erasure pattern $\be$. By Lemma \ref{lem:code_prop}, with probability at least \eqref{eq:prob:code_prop} the list $\mc{L}$ contains at most  $\log \log (n) \capgap/2$ codewords if $\wt(\bZ) < n^{3/4}$ and for each $(m,m')$ in the list there exists a set $V(m,m',k,k')$ of size at least $\frac{\capgap n}{4 (K^2 - K)} = O(\capgap n/\log^2 n)$. We union bound over all erasure patterns $\be$ and messages $m$ to show that with high probability the code construction satisfies these conditions. To complete the proof, note that by Lemma \ref{lem:zweight} the weight of $\bZ$ is such that the list size is at most $\log\log(n)\capgap/2$ with probability $1 - \exp(- n^{1/2}/2)$. Thus from Lemma \ref{lem:mldecoder}  decoding succeeds with probability $1 - \exp(c \epsilon n^{1/2} \log \log n / \log^3 n)$ for some $c > 0$. Therefore the probability of error goes to $0$ as $n \to \infty$, showing that $R$ is achievable.
This completes the proof of Theorem~\ref{th:stochastic}.

\section{Deterministic codes: Proof of Theorem~\ref{th:deterministic}}

In this section we show that the stochastic nature of our code design is essential. Specifically, we show that
any series of $(n, 2^{nR})$ {\em deterministic} codes $(\renc_n,\ddec_n)$ (i.e., for which $\renc_n : [2^{nR}] \to \cX^n$ depends only on $m \in [2^{nR}]$) that allow communication over our channel model with average probability of error $\e_n $ tending to zero must satisfy $R \leq1-2p$. An example illustrating our proof appears at the end of the section (see Figure \ref{fig:attack})


We show, by presenting an adversarial strategy, that for any constant $\delta >0$ and sufficiently large values of $n$, any deterministic code $(\renc_n,\ddec_n)$  with  $R=1-2p+\delta$ will have average error $\e_n=\Omega(1)$ (where $\e_n$ does not depend on $n$ but will depend on $\delta$).  The adversarial strategy follows the ``wait and push strategy'' (used in \cite{dey_improved_2012,dey_upper_2013} for the causal binary bit-flip channel and in \cite{bassily_causal_2014} for the erasure case) in which the adversary ``waits" a certain amount of time without performing any action, and then based on the information the adversary has seen so far ``pushes'' (i.e., corrupts) the transmitted codeword in a malicious manner causing a decoding error with some probability. 

For a given message $m$ and time parameter $\ell$, let $\renc_\ell(m)$ be the set of messages that have corresponding codewords that agree with $\renc(m)$ on the first $\ell$ entries.
The set $\renc_\ell(m)$ plays an important role in our analysis and will be referred to as the ``$\ell$-consistency'' set.
Notice that Calvin {\em cannot} construct $\renc_\ell(m)$ after $\ell$ bits of $\renc(m)$ have been transmitted (as, due to the delay, he has no knowledge of the $\ell$'th bit in  $\renc(m)$).
However, as the delay of Calvin is only 1-bit, at each time step $\ell$, Calvin will can construct two {\em potential} consistency sets.
The set $\renc^0_\ell(m)$ corresponding to the case that the $\ell$'th bit transmitted is 0 and one set $\renc^1_\ell(m)$ corresponding to the case that the bit is 1. It holds that 
$\renc^0_\ell(m) \cup \renc^1_\ell(m) = \renc_{\ell-1}(m)$.

We start by defining (and analyzing) the ``wait'' phase of Calvin.  We will then turn to discussing the push phase.

\subsection{``Wait'' phase}

In the wait phase Calvin proceeds as follows:
\begin{enumerate}
\item {\it (Wait-$1$):} Calvin starts by waiting until $(R-\delta)n=(1-2p+\delta)n + 1$ bits of the transmitted codeword are sent.
\item For each value of $\ell > (1-2p+\delta)n$, on transmission of the $\ell$'th bit of the transmitted codeword, Calvin constructs the sets $\renc^0_\ell(m)$ and $\renc^1_\ell(m)$.
Let $A_\ell = \max(|\renc^0_\ell(m)|,|\renc^1_\ell(m)|)$ and $a_\ell = \min(|\renc^0_\ell(m)|,|\renc^1_\ell(m)|)$. Clearly $A_\ell \geq a_\ell$. In addition, $A_\ell+a_\ell$ is exactly $\renc_{\ell-1}(m)$ and we will show shortly that with high probability over messages $m$ it holds that for $\ell=(1-2p+\delta)n$ the size $A_\ell+a_\ell$ is at least $2^{\Theta(n)}$. Based on the value of $A_\ell+a_\ell$ Calvin decides to either continue waiting or to move on to the push phase.
Specifically:
\begin{itemize}
\item {\it (Wait-$2$):} If $A_\ell+a_\ell$ is greater than $\delta' n$, Calvin does nothing and waits for the next bit to be transmitted. Here $\delta'=\delta /4$.
\item {\it (Attack):} If $A_\ell+a_\ell$ is less than $\delta' n$ but at least as large as $\frac{c}{\delta}$ Calvin, sets the ``transition time'' $\ell^*$ to equal the current value of $\ell$, stops the ``wait'' phase, and moves on to the ``push'' phase to be discussed below in Section~\ref{sec:push} in detail. 
\item {\it (Error):} If $A_\ell+a_\ell$ is less than $\frac{c}{\delta}$, set the ``transition time'' $\ell^*$ to equal the current value of $\ell$ and declare an error of Type $1$.
\end{itemize}
\end{enumerate}
By our definitions, Calvin either declares an error or will move on to the push phase at some point in time $\ell^*$.
For the latter we say that the transition to the push phase is successful for message $m$, namely the size $A_{\ell^*}+a_{\ell^*} < \delta' n$ is at least $\frac{c}{\delta}$ for a sufficiently large constant $c$ to be determined shortly. Otherwise we say that the transition has failed (there is an error of Type $1$).
We now show that with some constant probability over messages $m$, the transition to the push phase is successful (no error of Type $1$).

\begin{lemma}
\label{lemma:wait}
Let $c$ be a sufficiently large constant to be determined shortly.
Let $n$ be sufficiently large.
Let $\ell^*$ be the first point in time for which $A_{\ell^*}+a_{\ell^*} < \delta' n$.
With probability at least  $2^{-8c/\delta^2}$ over messages $m$, it holds that $A_{\ell^*}+a_{\ell^*}$ is of size at least $\frac{c}{\delta}$.
\end{lemma}

\begin{proof}
We first note that in \cite{langberg_binary_2009} it is shown, using the pigeonhole principle, that the probability over messages $m$ that for $\ell=(1-2p+\delta)n$ the size of $\renc_\ell(m)$ (and thus $A_\ell+a_\ell$) is at least $2^{\delta n/2}$ is at least $1-2^{-\delta n/2}$.
Let $E_1$ be the event that  Alice chooses a message $m$ for which the corresponding consistency set  $\renc_\ell(m)$ is of size at least $2^{\delta n/2}$.

We now address the probability, given $E_1$ that the transition of Calvin to the push phase has failed.
This can happen for messages $m$ only if $\renc_{\ell^*-2}(m)=A_{\ell^*-1}+a_{\ell^*-1} \geq \delta' n$ and  $\renc_{\ell^*-1}(m)=A_{\ell^*}+a_{\ell^*}$ is of size less than $\frac{c}{\delta}$. 
Or in other words, failure happens only for messages $m$ that at some point in time have consecutive consistency sets of sizes that {\em jump} from above $\delta' n$ to below $\frac{c}{\delta}$.

Consider a codeword chosen uniformly at random from the codebook of Alice (this corresponds to choosing a uniformly distributed message $m$). 
One may expose this codeword bit by bit according to the conditional probability given the choices made thus far. 
In such a process for time parameter $\ell$, if $A_{\ell-1}+a_{\ell-1}$ is at least $\delta' n$ and the value of $A_\ell + a_\ell$ is less than $\frac{c}{\delta}$, there will be a failure for Calvin with probability 
	\begin{align}
\frac{A_{\ell}+a_{\ell}}{A_{\ell-1}+a_{\ell-1}}.
	\end{align}
Notice that $A_\ell+a_\ell$ is equal to either $A_{\ell-1}$ or $a_{\ell-1}$ by our exposure process. Moreover, for sufficiently large $n$, $A_\ell+a_\ell=a_{\ell-1}$ as  otherwise $A_{\ell-1} = A_\ell+a_\ell < \frac{c}{\delta}$ which in turn implies that $A_{\ell-1}+a_{\ell-1} \leq \frac{2c}{\delta}$ in contradiction to $A_{\ell-1} + a_{\ell-1} \geq  \delta' n$. 
This implies that in such cases, the conditional probability of error at time $\ell$ is
	\begin{align}
\frac{a_{\ell-1}}{A_{\ell-1}+a_{\ell-1}},
	\end{align}
%
%
or equivalently, in such cases the conditional  probability that the exposure process does {\em not} induce a failed transition is 
	\begin{align}
1-\frac{a_{\ell-1}}{A_{\ell-1}+a_{\ell-1}}.
	\end{align}
We conclude that the probability $q$ over codewords (i.e., messages $m$) that the transition is successful for Calvin is
	\begin{align}
q=\prod_{\ell} \left(1-\frac{a_{\ell-1}}{A_{\ell-1}+a_{\ell-1}} \right),
	\end{align}
where the product is over $\ell$ for which (as specified above) $A_{\ell-1}+a_{\ell-1} \geq \delta' n$ and $ 1 \leq a_{\ell-1} \leq \frac{c}{\delta}$.
As there can be at most $n$ such values of $\ell$ we have that 
	\begin{align}
q \geq \prod_{k=1}^n \left(1-\frac{c}{\delta x_{k}} \right),
	\end{align}
where $x_{\ell}$ is a strictly decreasing sequence of integers greater than $\delta' n$.
It now holds that the setting for which our lower bound on $q$ is minimum is that in which $x_k$ are consecutive integers (in increasing order) starting from $x_{n}=\delta' n+1$.
I.e. $x_k = \delta' n + k$.
We conclude that (for sufficiently large values of $n$) $q$ is bounded from below by $e^{-\left ( \frac{4c}{\delta^2} \right )}$\footnote{A tighter analysis indicates a better lower bound of $\left( \frac{\delta}{8}\right)^{\frac{c}{\delta}}$ -- due to the intricacy of this analysis we omit it here.}:
\begin{align*}
q & \geq \prod_{k = 1}^{n}\left(1-\frac{c}{\delta (\delta' n + k)} \right) \geq \prod_{k = 1}^{n}\left(1-\frac{c}{\delta \delta' n } \right) =  \left(1-\frac{4c}{\delta^2 n } \right)^n \geq e^{-\left ( \frac{4c}{\delta^2} \right )}.
\end{align*}

All in all, using the union bound with event $E_1$, for sufficiently large $n$ we have with probability at
least $e^{-\left ( \frac{4c}{\delta^2} \right )} - 2^{-\delta n/2} \geq
2^{-\left ( \frac{8c}{\delta^2} \right )}$ that Calvin's transition to the push phase will result in
a success.  \end{proof}


\subsection{``Push'' phase}\label{sec:push}
Calvin's corrupting algorithm now proceeds as follows. 
\begin{enumerate}
\item Calvin chooses a ``plausible transmission'' $\bX'$ uniformly at random from $\renc^0_{\ell^*}(m) \cup \renc^1_{\ell^*}(m)$.
\item For each value of $\ell \geq \ell^*$, either $|\renc^0_\ell(m)|\geq 1$ and
$|\renc^1_\ell(m)| \geq 1$ (Calvin has uncertainty about $X_\ell$, since it
is possible for $X_\ell$ to equal either $0$ or $1$), or $|\renc^i_\ell(m)|
=0$ for some $i \in \{0,1\}$ (Calvin is certain about
$X_\ell$, since all surviving codewords have $\bX_\ell = 1-i$).  Calvin
does the following:
\begin{enumerate}
\item {\it (Calvin uncertain about $\bX_\ell$):} If $|\renc^0_\ell(m)|\geq 1$ and $|\renc^1_\ell(m)| \geq 1$, then Calvin erases $\bX_\ell$.
\item {\it (Calvin certain about $\bX_\ell$):} If $|\renc^i_\ell(m)| =0$
for some $i$, 
\begin{enumerate}
\item  {\it (Erasing disambiguating information):} If $\bX'_\ell = i$ and hence $\bX_\ell \neq \bX'_\ell$, Calvin erases $\bX_\ell$.
\item {\it (No action):} If $\bX'_\ell = 1-i$ and hence $\bX_\ell = \bX'_\ell$, Calvin does not erase $\bX_\ell$.
\end{enumerate}
\end{enumerate}
\end{enumerate}
 
If Calvin can successfully continue the above process until the end without
violating his total erasure budget, then clearly both the codewords
$\bX$ and $\bX^\prime$ are consistent with the vector received by Bob as
all the indices where they differ are erased by Calvin.
We will now argue that Calvin can indeed complete this process, that is,
the total number of erasures required is $\leq pn$.

Erasures are introduced by Calvin in steps 2(a) and 2(b)ii. 
Let us consider the full binary tree of depth $n$ where the edges are 
labeled by $0$ and $1$, and let us consider the code as its subtree.
Here, each path from the root to a leaf represents the codeword that
is composed of the bits labeling the branches along that path.
From Calvin's perspective, at the beginning of the push phase, the encoder
state is either of the two nodes representing the subsequences $\bX^{\ell^*-1}0$
and $\bX^{\ell^*-1}1$. All the paths via these two nodes represent the two
sets of codewords $\renc^0_{\ell^*}(m)$ and $\renc^1_{\ell^*}(m)$ respectively.
Since the total number of such paths is at most $\delta' n=\delta n/4$, 
there are at
most $\delta n/4$ branchings in the subtree rooted at the node corresponding to 
$\renc_{\ell^*-1}(m)$ (i.e., the subtree spanning the codewords in $\renc^0_{\ell^*}(m)$ and $\renc^1_{\ell^*}(m)$).
This implies, that along any path in this subtree Calvin will encounter at most $\delta n/4$ branching nodes.
In other words, Calvin will
encounter step 2(a) at most $\delta n/4$ times. This upper bounds the total 
number of erasures due to step 2(a) by $\delta n/4$.

Counting the number of required erasures in step 2(b)ii can be done following similar
analysis as in~\cite{langberg_binary_2009}, using the Plotkin bound and Turan's theorem.
We briefly reprise the analysis here. We consider the codebook of length 
$n' \leq (2p-\delta )n$ formed by the completions of $\bX^{\ell^*-1}$.
Let us consider the graph with these codewords as nodes, and two codewords
connected if their Hamming distance is at most $d=pn-\delta n/4$.
By the Plotkin bound, any independent set in this graph has at most 
$4p/\delta$ nodes. This implies, by Turan's theorem, that the average degree
$\Delta$ and the number of nodes $|V|$ satisfy
\begin{align*}
& \frac{\Delta +1}{|V|} \geq \frac{\delta }{4p}.
\end{align*}
Implying that
\begin{align*}
\frac{\Delta }{|V|} \geq \frac{\delta }{8p}.
\end{align*}
Thus the probability of two randomly chosen codewords being at
a distance at most $d$ is 
\begin{align*}
\frac{|{\cal E}|}{|V|^2}& =\frac{\Delta |{\cal V}|}{2|V|^2}
 \geq \frac{\delta }{16p}.
\end{align*}
So with this constant probability, Calvin's remaining erasure budget
$pn-\delta n/4$ is sufficient to erase (in step 2(b)ii) all the positions where 
$\bX$ and $\bX'$ differ.

All in all, the success probability of Calvin is bounded by below by his success in the wait phase times that in the push phase which is a constant independent of $n$: 
	\begin{align}
\frac{\delta }{16p} \cdot 2^{-\left ( \frac{8c}{\delta^2} \right )} \geq 2^{-\left ( \frac{16c}{\delta^2} \right )}.
	\end{align}



\begin{remark}
We note that the proof of Theorem~\ref{th:deterministic} does not hold for stochastic codes. While the wait phase may have an analogous analysis that fits the stochastic setting, the push phase breaks down. Specifically, a crucial part of the push phase is step 2(a) which erases any location in which there is some uncertainty on behalf of Calvin regarding the current symbol. In deterministic codes, step 2(a) may occur in only few locations, whereas  in the stochastic setting the number of branchings in the subtree rooted at the node corresponding to Calvin's view so far may be large, and thus step 2(a) may be too costly. Indeed, in our code design for the achievability proof presented in Section \ref{sec:lower}, each and every location includes a branching point.
\end{remark}

\subsection{Illustration of proof for Theorem~\ref{th:deterministic}}

\begin{figure}
\begin{center}
\includegraphics[scale=0.4]{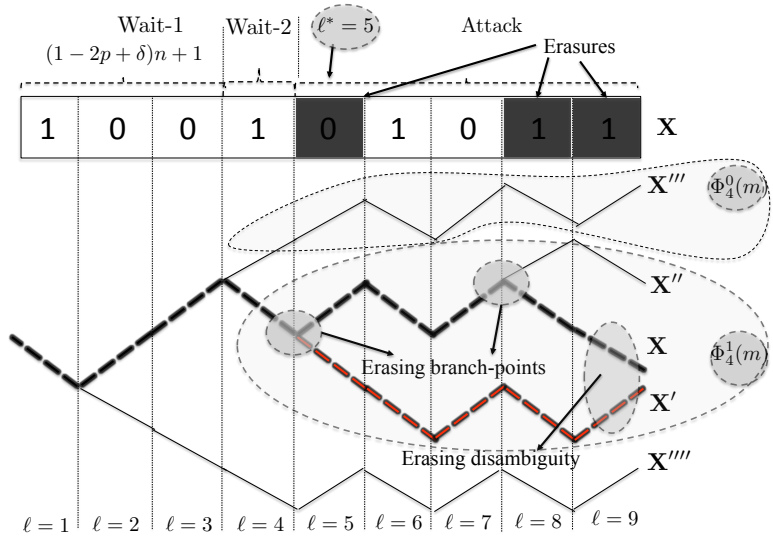}
\caption{Illustration of proof for Theorem \ref{th:deterministic}}
\label{fig:attack}
\end{center}
\end{figure}

In Figure \ref{fig:attack}, we demonstrate a toy example showing an adversarial attack against  a deterministic code of block-length $n = 9$, comprising of $5$ codewords $\bX, \bX', \bX'', \bX'''$ and $\bX''''$, and hence the rate $R$ of this code equals $\log_2(5)/9 \approx .258$. In this toy-example, $p = 4/9$ (so at most $4$ erasures are possible over the length-$9$ transmission, though in this example only $3$ bits are actually erased), and hence the claim of Theorem~\ref{th:deterministic} is that for sufficiently large $n$, no rate asymptotically larger than $1-2p = 1/9$ is achievable, implying that ``not too many more'' than $2$ messages can be reliably transmitted via a deterministic code. In particular, this example aims to show that for the specific code shown, the $5$ messages corresponding to the $5$ codewords chosen cannot be reliably transmitted. (To keep the example dimensions manageable, not all parameters in the example match those in our proofs -- in particular, no suitable value of the ``rate-excess parameter $\delta$ exists that matches those required by our proofs, for the ``small'' value of $n$ chosen.) 
The zig-zag lines at the bottom of  Figure \ref{fig:attack} show the ``code-tree'', the binary tree representing the $5$ length-$9$ codewords as paths in an (incomplete) depth-$9$ binary tree -- segments angled upwards indicate $0$'s in that location, and segments angled downwards indicate $1$'s in that location, and hence the five codewords are respectively $\bX = 100101011$, $\bX' = 100111010$, $\bX'' = 1001010001$, $\bX''' = 100001010$, and $\bX'''' = 111101010$. 

The codeword actually transmitted, $\bX$, is shown as the black-shaded path in the code-tree. Note that all codewords in the code have the same first bit ($1$), and hence Calvin has to wait until $\ell = 2$ before he knows that the transmitted codeword is not $\bX''''$. In general, Calvin's initial two phases are ``wait'' phases, in which he does not erase any bits. Specifically, Calvin is {\it always} in the ``Wait-$1$'' phase for {\it exactly} the first $(1-2p+\delta)n+1$ bits (shown in this figure as the first 3 bits), and then he continues waiting in the ``Wait-2'' phase until time $\ell^\ast$ (when the number of codewords consistent with his observations up to the time $\ell^\ast - 1$ is somewhere in the range $(c/\delta,\delta n)$, for some constant $c$ specified in Theorem~\ref{th:deterministic}). In this example the Wait-$2$ phase is of length $1$, since at time $\ell^\ast = 5$ Calvin observes the $4$th transmitted bit $x_4 = 1$, and realizes that the transmitted codeword does not equal $\bX'''$ -- hence his consistency set at this time (denoted $\renc_{4}(m)$) shrinks to become $\{\bX,\bX',\bX''\}$, and is of size $3$. At this point, Calvin segues to the ``Attack'' phase. Specifically, he first chooses a random codeword from his consistency 
set $\Phi_4(m)$ (in this example $\bX'$, denoted as the shared black-red path in the code-tree), and tries to confuse Bob between $\bX$ and $\bX'$. Specifically, whenever Calvin sees a ``branch-point'', {\it i.e.}, a location $\ell$ in which Alice may have transmitted either a $0$ or a $1$ ({\it i.e.}, in which there are codewords corresponding to both $\Phi_{\ell}^0(m)$ and $\Phi_{\ell}^1(m)$) he erases the 
corresponding bit (Step 2(a) in the push phase of Theorem \ref{th:deterministic}).
This happens in at most $\renc_{\ell^*-1}(m) < \delta n$ locations and thereby denies Bob knowledge of the value of these bits of $\bX_m$ (in this example, there are branch-points at $\ell = 5$ and $\ell = 8$). 
Also, if there is no branch-point, but $\bX_\ell \neq \bX'_\ell$ (such a bit would enable Bob to disambiguate between $\bX$ and $\bX'$), 
Calvin erases such bits as well (this happens at $\ell = 9$). At the end, both $\bX$ and $\bX'$ are equally likely from Bob's perspective. Care is required to ensure that Calvin does not run out of erasures in the Attack phase -- this is analyzed in Theorem \ref{th:deterministic} in detail.

\section{Omniscient adversary: stochastic vs. deterministic encoding}
\label{sec:stochdet}

\newcommand{\cE}{{\cal E}}
\newcommand{\cA}{{\cal A}}
\newcommand{\cM}{{\cal M}}

Here we argue that for a bit-flipping adversary who can flip
upto $p$ fraction of bits in a codeword, the capacity under stochastic
encoding is the same as that under deterministic encoding.

Suppose a rate $r$ is achievable under stochastic encoding and average
error probability. Let us suppose that there is a sequence of stochastic
codes achieving average probability of error $\epsilon_n$ for length $n$,
such that $\epsilon_n \rightarrow 0$. Let us now consider a fixed $n$,
and let $\cA_m$ denote the set of vectors for which the decoder outputs
the message $m$. Let $\cE:=\{\hat{M}\neq M\}$.
\begin{align*}
\cM_n :=\{m: Pr\{\cE|M=m\}< 1\}.
\end{align*}
For each message $m\in \cM_n$, there is a `good' codeword $\bX(m)$ such that
$Pr \{\bX (m) |M=m\} >0$, and 
the adversary does not have the power to move it outside $\cA_m$.
In other words, the ball of radius $pn$ around $\bX(m)$ is completely
contained in $\cA_m$. 
Let $\alpha_n :=Pr\{M\in \cM_n\}$. Clearly,
\begin{align*}
\epsilon_n & = \alpha_n Pr\{\cE|M\in \cM_n\} + (1-\alpha_n) Pr\{\cE|M\not\in\cM_n\}\\
& \geq 0 + (1-\alpha_n)\\
& \geq 1-\alpha_n.
\end{align*}
Let $\cC_n:= \{\bX(m):m\in \cM_n\}$ be a set of good codewords for messages
in $\cM_n$. We now argue that the sequence of deterministic codes
$\cC_n$ with decoder decision regions $\cA_m; m\in \cM_n$ have zero
error, and an asymptotic rate $r$.
That the code has zero error probability follows because, for each codeword
$\bX(m)$, the adversary does not have the power to move it outside $\cA_m$.
\begin{align*}
& H(M) \geq nr\\
\Rightarrow & H(\alpha_n) + \alpha_n H(M|M\in \cM_n) 
+ (1-\alpha_n)H(M|M\not\in \cM_n) \geq nr\\
\Rightarrow & \alpha_n H(M|M\in \cM_n) \geq nr -1 -(1-\alpha_n) H(M|M\in \cM_n)\\
\Rightarrow & \alpha_n H(M|M\in \cM_n) \geq nr -1 -\epsilon_n H(M|M\in \cM_n)\\
\Rightarrow & \alpha_n H(M|M\in \cM_n) \geq nr -1 -\epsilon_n nr\\
\Rightarrow & \frac{1}{n} H(M|M\in \cM_n) \geq \frac{1}{\alpha_n}\left((1-\epsilon_n) r -\frac{1}{n}\right)\\
\Rightarrow & \frac{1}{n} \log_2 |\cM_n| \geq \frac{1}{\alpha_n}\left((1-\epsilon_n) r -\frac{1}{n}\right)
\end{align*}
Since $\epsilon_n \rightarrow 0$ and $\alpha_n \rightarrow 1$, the rate
of $\cC_n$ converges to $r$.

\appendix

\bibliographystyle{unsrt}
\bibliography{../../mainfunction.bib}

\end{document}